\documentclass[graybox]{svmult}

\usepackage{mathptmx}       % selects Times Roman as basic font
\usepackage{helvet}         % selects Helvetica as sans-serif font
\usepackage{courier}        % selects Courier as typewriter font
\usepackage{type1cm}        % activate if the above 3 fonts are
\usepackage{makeidx}         % allows index generation
\usepackage{graphicx}        % standard LaTeX graphics tool
\usepackage{multicol}        % used for the two-column index
\usepackage[bottom]{footmisc}% places footnotes at page bottom

\usepackage{amsmath,amsfonts,amsbsy}
\usepackage{dsfont}

\usepackage{tikz}
\usetikzlibrary{arrows}

\newcommand{\eps}{\varepsilon}

\newcommand{\Id}{\mathds{1}}
\newcommand{\B}{\mathbb{B}}
\newcommand{\Beff}{\mathbb{B}\sub{eff}}
\newcommand{\C}{\mathbb{C}}
\newcommand{\R}{\mathbb{R}}
\newcommand{\Z}{\mathbb{Z}}

\newcommand{\T}{\mathbb{T}}

\newcommand{\Hf}{\mathcal{H}\sub{f}}
\newcommand{\U}{\mathrm{U}}
\newcommand{\A}{\mathcal{A}}

\newcommand{\scal}[2]{\left\langle #1, #2 \right\rangle}
\newcommand{\norm}[1]{\left\| #1 \right\|}

\newcommand{\eu}{\mathrm{e}}
\newcommand{\iu}{\mathrm{i}}
\newcommand{\di}{\mathrm{d}}
\newcommand{\act}{\triangleleft}
\newcommand{\sub}[1]{_{\text{#1}}}
\newcommand{\tra}{^{\mathsf{T}}}

\DeclareMathOperator{\Tr}{Tr}

\DeclareMathOperator{\Ran}{Ran}

\DeclareMathOperator{\Span}{Span}

\newcommand{\ie}{{\sl i.\,e.\ }}
\newcommand{\eg}{{\sl e.\,g.\ }}
\newcommand{\set}[1]{ \left\{  #1 \right\}}

\makeindex             % used for the subject index
\begin{document}

\title*{Chern and Fu--Kane--Mele invariants as topological obstructions}
\author{Domenico Monaco}
\institute{Domenico Monaco \at Eberhard Karls Universit\"{a}t T\"{u}bingen, Auf der Morgenstelle 10, 72076 T\"{u}bingen (Germany)\newline
\email{domenico.monaco@uni-tuebingen.de}}
\maketitle

%\abstract*{Each chapter should be preceded by an abstract (10--15 lines long) that summarizes the content. The abstract will appear \textit{online} at \url{www.SpringerLink.com} and be available with unrestricted access. This allows unregistered users to read the abstract as a teaser for the complete chapter. As a general rule the abstracts will not appear in the printed version of your book unless it is the style of your particular book or that of the series to which your book belongs.
%Please use the 'starred' version of the new Springer \texttt{abstract} command for typesetting the text of the online abstracts (cf. source file of this chapter template \texttt{abstract}) and include them with the source files of your manuscript. Use the plain \texttt{abstract} command if the abstract is also to appear in the printed version of the book.}

\abstract{%
The use of topological invariants to describe geometric phases of quantum matter has become an essential tool in modern solid state physics. The first instance of this paradigmatic trend can be traced to the study of the quantum Hall effect, in which the Chern number underlies the quantization of the transverse Hall conductivity. More recently, in the framework of time-reversal symmetric topological insulators and quantum spin Hall systems, a new topological classification has been proposed by Fu, Kane and Mele, where the label takes value in $\Z_2$. \newline
We illustrate how both the Chern number $c \in \Z$ and the Fu--Kane--Mele invariant $\delta \in \Z_2$ of $2$-dimensional topological insulators can be characterized as topological obstructions. Indeed, $c$ quantifies the obstruction to the existence of a frame of Bloch states for the crystal which is both continuous and periodic with respect to the crystal momentum. Instead, $\delta$ measures the possibility to impose a further time-reversal symmetry constraint on the Bloch frame. 
}%


\section{Introduction}
\label{sec:Intro}

One of the most prominent instances of Wigner's ``unresonable effectiveness of mathematics'' in condensed matter systems is provided by \emph{topological insulators} \cite{HasanKane10}. These materials, although insulating in the bulk, have the property of conducting currents on their boundary, making them amenable to various types of applications in material science, and even in quantum computing. A thorough understanding of the transport properties of these materials, however, can be achieved only by investigating the topology of the occupied states that fill the bulk energy bands, by virtue of a principle known as the bulk-edge correspondence. Consequently, some of the techniques of topology and differential geometry, once relegated to abstract mathematics, have nowadays become common knowledge also among solid state physicists.

To better understand how topology enters in the world of condensed matter systems, it is particularly instructive to consider the archetypal example of a topological insulator, given by a \emph{quantum Hall system} \cite{Graf07}. An effectively 2-dimensional crystalline medium is immersed in a uniform magnetic field perpendicular to the plane of the sample, and an electric current is driven in one direction along the crystal. The induced current is measured in the transverse direction. In a remarkable experiment, performed at very low temperatures by von Klitzing and his collaborators \cite{vonKlitzing80}, the (Hall) conductivity $\sigma_{\mathrm{H}}$ associated to this transverse current was shown to display {\it plateaux} which occurred at integer multiples of a fundamental constant, measured morever with an astounding precision:
\begin{equation} \label{eqn:Hall}
\sigma_{\mathrm{H}} = n \, \frac{e^2}{h}, \quad n \in \mathbb{Z}.
\end{equation}
Later theoretical investigations showed that a topological phenomenon underlies this quantization: the integer $n$ in the above formula was shown to be the first Chern number of a vector bundle, naturally associated to the quantum system \cite{ThoulessKohmotoNightingaledenNijs82, AvronSeilerSimon94, BellissardvanHelstSchulz-Baldes94}.

The only role played by the magnetic field in quantum Hall systems is that of breaking \emph{time-reversal symmetry}: if the system were time-reversal symmetric, then the Hall conductivity would vanish, and the system would remain in an insulating state. This fact was clarified by Haldane \cite{Haldane88}, who showed that non-trivial topological phases can be displayed also in absence of a magnetic field, thus initiating the field of \emph{Chern insulators} \cite{Bestwick15, Chang15}. Picking up on the work by Haldane, Fu, Kane and Mele \cite{KaneMele05, FuKane06, FuKaneMele07} later introduced a model which still displays a topological phase even if time-reversal symmetry is preveserved, and is by now recognized as a milestone in the history of topological insulators. The phenomenon that the model proposed to illustrate is that of the \emph{quantum spin Hall effect}, which differs from the quantum Hall effect in that the external magnetic field is replaced by spin-orbit interactions (exactly to preserve time-reversal symmetry), and spin rather than charge currents flow on the boundary of the sample. From the point of view of topological phases, the peculiarity of this phenomenon is that, contrary to what happens for Chern and quantum Hall insulators, one can only distinguish between the trivial (insulating) and non-trivial (quantum spin Hall) phase: the label is then assigned by a $\mathbb{Z}_2$-valued topological index. Giving a full account of the geometric nature of this invariant has been a primary objective for mathematical physicists in the last decade, and a plethora of mathematical tools has been used in this endeavour, ranging from $K$-theory to homotopy theory, from functional analysis to noncommutative geometry, from equivariant cohomology to operator theory. We refer to \cite{FiorenzaMonacoPanati16_F, ProdanSchulz-Baldes16, CorneanMonacoTeufel16} for recent accounts on the ever-growing literature on the subject.

The purpose of this contribution is to express both the Chern number and the Fu--Kane--Mele $\Z_2$ index of 2-dimensional topological insulators in a common framework, provided by \emph{obstruction theory}. It will be shown how both invariants arise as \emph{topological obstructions} to the existence of a Bloch frame, which roughly speaking can be described as a set of continuous functions which parametrize the occupied states of the physical systems and are compatible with its symmetries, namely periodicity with respect to the Bravais lattice of the crystal and, possibly, time-reversal symmetry; a precise definition will be given in the next Section. The nature of these topological invariants as obstructions was early realized \cite{Kohmoto85, FuKane06}, employing methods from bundle theory and using local trivializing charts. Our strategy relies instead on successive extensions of the definition of the Bloch frame, which is well-suited for induction on the dimension of the system and is reminiscent of the extension of a section of a bundle along the cellular decomposition of its base space. We use only basic facts from linear algebra and the topology of the group of unitary matrices $\U(m)$; besides, our method has the further advantage of constructing the required Bloch frame in an algorithmic fashion.

\section{Topology of crystalline systems}

\subsection{Periodic Hamiltonians}

To set up a rigorous investigation of topological phases of quantum matter, we first have to understand the mathematical description of crystalline systems. The starting point is a \emph{periodic Hamiltonian}: one could think of continuous models described by Schr\"odinger operators, or of discrete, tight-binding models described by hopping matrices. Periodicity means that the operator $H$ should commute with the translations associated to a lattice $\Gamma \simeq \mathbb{Z}^d \subset \mathbb{R}^d$, namely the Bravais lattice of the crystal under scrutiny. This symmetry of the Hamiltonian leads to a partial diagonalization of it, by looking at common (generalized) eigenstates for the Hamiltonian and the translations: this procedure, which is reminiscent of the Fourier decomposition, goes by the name of \emph{Bloch-Floquet reduction} \cite{MonacoPanati15}. In this representation, the Hamiltonian becomes a fibered operator, with fibre $H(k)$ acting on a space $\mathcal{H}_\mathrm{f}$ containing the degrees of freedom associated to a unit cell for $\Gamma$. The parameter $k \in \R^d$, also called \emph{crystal} or \emph{Bloch momentum}, is determined up to translations by vectors in the dual lattice $\Lambda := \Gamma^*$, and thus can be considered as an element of the \emph{Brillouin torus} $\mathbb{T}^d := \mathbb{R}^d / \Lambda$. Indeed, the fibre Hamiltonians at $k$ and $k + \lambda$, $\lambda \in \Lambda$, are unitarily intertwined by a representation $\tau \colon \Lambda \to \mathcal{U}(\mathcal{H}_\mathrm{f})$, namely
\[ H(k+\lambda) = \tau_\lambda \, H(k) \, \tau_\lambda^{-1}. \]
The above relation will be called \emph{$\tau$-covariance} in what follows.

Due to the compactness of the unit cell, under fairly general assumptions%
\footnote{In continuous models, where $H$ is a Schr\"{o}dinger operator, these assumptions usually amount to asking that the electromagnetic potentials be infinitesimally Kato-small (possibly in the sense of quadratic forms) with respect to the kinetic part \cite{ReedSimon4}.} %
the operator $H(k)$ has discrete spectrum: the function $k \mapsto E_n(k)$, associated to one of its eigenvalues (labelled, say, in increasing order), is called the \emph{Bloch band}. The spectrum of the original Hamiltonian is recovered by considering the (possibly overlapping) ranges of all these functions, and leads to the well-known band-gap structure of the spectrum of a periodic operator. If one assumes that the Fermi energy of the system lies in a spectral gap for $H$, then it makes sense to consider the Fermi projector $P(k)$ on the $m$ occupied bands. The gap condition implies that the dependence of $P(k)$ on $k$ is analytic, and the family of operators $P(k)$ is also $\tau$-covariant (see \eg \cite[Prop.~2.1]{PanatiPisante13}).

For the applications to topological insulators that we are aiming at, we need to consider also a further symmetry of the Hamiltonian, namely 	\emph{time-reversal symmetry}. This is implemented antiunitarily on the Hilbert space of the quantum particle, and flips the arrow of time (and hence the crystal momentum). Mathematically, this amounts to require the existence of an antiunitary operator $\Theta$ on $\Hf$, squaring to $\pm \mathbf{1}_{\Hf}$, and such that
\[ H(-k) = \Theta \, H(k) \, \Theta^{-1}. \]
We say that the family of operators $H(k)$ is \emph{time-reversal symmetric} if the above holds. It is easy to verify that the Fermi projectors associated to a time-reversal symmetric Hamiltonians are time-reversal symmetric as well. In what follows, we will focus mainly on the case of a \emph{fermionic} time-reversal symmetry operator, namely on the case where $\Theta^2 = - \mathbf{1}_{\Hf}$, as is the case for example for quantum spin Hall systems.

\subsection{Bloch bundle, Berry connection and Berry curvature}
\label{sec:Berry}

From the previous analysis of periodic and time-reversal symmetric Hamiltonians, we ended up with a family of projectors $\{P(k)\}_{k \in \mathbb{R}^d} \subset \mathcal{B}(\Hf)$, $P(k)^* = P(k) = P(k)^2$, satisfying the following properties:
\begin{description}
 \item[(P$_1$)] \emph{analyticity}: the map $k \mapsto P(k)$ is a real-analytic map on $\R^d$ with values in $\mathcal{B}(\Hf)$;
 \item[(P$_2$)] \emph{$\tau$-covariance}: the map $k \mapsto P(k)$ satisfies
 \[ P(k+\lambda) = \tau_\lambda \, P(k) \, \tau_\lambda^{-1} \]
 for a unitary representation $\tau \colon \Lambda \to \mathcal{U}(\Hf)$ of the lattice $\Lambda \simeq \Z^d \subset \R^d$;
 \item[(P$_3$)] \emph{time-reversal symmetry}: the map $k \mapsto P(k)$ satisfies
 \[ P(-k) = \Theta \, P(k) \, \Theta^{-1} \]
 for an antiunitary operator $\Theta \colon \Hf \to \Hf$ such that $\Theta^2 = - \mathbf{1}_{\Hf}$.
\end{description}

The topology underlying the quantum system described by the Hamiltonian $H$ is encoded in its eigenprojectors, satisfying the above properties%
\footnote{In order for (P$_2$) and (P$_3$) to be compatible with each other, one should also require that $\tau_\lambda \, \Theta = \tau_\lambda^{-1} \, \Theta$ for all $\lambda \in \Lambda$. We will assume this in the following.}%
. Indeed, one can associate to any family of projectors satisfying (P$_1$) and (P$_2$) a vector bundle $\mathcal{E}$ over the torus $\mathbb{T}^d$, called the \emph{Bloch bundle}, via a procedure reminescent of the Serre--Swan construction: the fibre of $\mathcal{E}$ over the point $k \in \mathbb{T}^d$ is the $m$-dimensional vector space $\Ran P(k)$ (we refer to \cite{Panati07, MonacoPanati15} for details). The geometry of the Bloch bundle for $d=2$ is what enters in the theoretical understanding of the quantum Hall effect: the integer $n$ that equals the Hall conductivity \eqref{eqn:Hall} in natural units is the (\emph{first}) \emph{Chern number} of $\mathcal{E}$, defined as
\begin{equation} \label{eqn:c1}
c_1(P) := \frac{1}{2 \pi \mathrm{i}} \int_{\mathbb{T}^2} \Tr_{\Hf} \, \left( P(k) \, [\partial_1 P(k), \partial_2 P(k)] \right) \, \mathrm{d} k_1 \mathrm{d} k_2 \quad \in \mathbb{Z}.
\end{equation}
When $d = 2$, the above integer characterizes the isomorphism class of $\mathcal{E}$ as a vector bundle over $\T^2$ \cite{Panati07}. Since both quantum Hall and quantum spin Hall systems are 2-dimensional, in the following we will mostly restrict ourselves to $d=2$, where in particular the previous characterization holds.

In the case where $\set{P(k)}_{k \in \R^d}$ satisfies also (P$_3$), then the Bloch bundle can be equipped with further structure, namely that of a fiberwise antilinear endomorphism $\widehat{\Theta} \colon \mathcal{E} \to \mathcal{E}$, lifting the involution $\theta(k) = -k$ on the base torus and squaring to the operator which multiplies fiberwise by $-1$. We call a vector bundle endowed with such an endomorphism $\widehat{\Theta}$ a \emph{time-reversal symmetric vector bundle}. One can verify that if $d=2$ every such vector bundle is trivial, \ie isomorphic to the product bundle $\T^2 \times \C^m$, since under (P$_3$) the integrand in the definition \eqref{eqn:c1} of the Chern number is an odd function of $k$, and hence integrates to zero on $\T^2$ \cite{Panati07, MonacoPanati15}. However, the Bloch bundle may still be non-trivial as time-reversal symmetric bundle \cite{DeNittisGomi15, FiorenzaMonacoPanati16_F}. The index that characterizes the isomorphism class of $\mathcal{E}$ is the \emph{Fu--Kane--Mele index} $\delta(P) \in \Z_2$, first introduced in \cite{FuKane06} to describe quantum spin Hall systems. The expression of the $\Z_2$ index is slightly more involved than the one for the Chern number, and requires the introduction of some further terminology, which will be however essential in what follows. 

Given a family of projectors $\set{P(k)}_{k \in \R^d}$ of constant rank $m$, a \emph{Bloch frame} for it is a family of $m$-tuples of vectors $\Psi = \set{\psi_a(k)}_{1 \le a \le m, \: k \in \R^d}$, which are orthonormal and span the vector subspace $\Ran P(k) \subset \Hf$ for all $k \in \R^d$. If $P(k)$ depends smoothly on $k$, then the same can be required of the frame $\Psi$. We immediately stress that, when $\set{P(k)}_{k \in \R^d}$ satisfies (P$_1$) and (P$_2$), then a Bloch frame is nothing but a trivializing frame for the associated Bloch bundle, and hence the existence of a continuous frame is in general guaranteed only \emph{locally} in $k$. Let us also point out that, whenever a Bloch frame $\Psi$ exists (say on an open domain $\Omega \subset \R^d$), then any other Bloch frame $\Phi$ is obtained by setting
\begin{equation} \label{BlochGauge}
\phi_b(k) := \sum_{a=1}^{m} \psi_a(k) \, U(k)_{ab}, \quad 1 \le b \le m,
\end{equation}
where $U(k)$, $k \in \Omega$, is a unitary matrix, called the \emph{Bloch gauge}. We use the shorthand notation
\begin{equation} \label{act}
\Phi(k) = \Psi(k) \act U(k), \quad k \in \Omega,
\end{equation}
to write \eqref{BlochGauge} in a more compact form. This defines a free right action of $\U(m)$ on frames, meaning that $(\Psi \act U_1) \act U_2 = \Psi \act (U_1 \, U_2)$ and that $\Psi \act U_1 = \Psi \act U_2$ if and only if $U_1 = U_2$.

When a (local) Bloch frame $\Psi = \set{\psi_a(k)}_{1 \le a \le m, \: k \in \R^d}$ is given, then one can define the \emph{Berry connection}, \ie the matrix-valued $1$-form given by 
\begin{equation} \label{Berry}
A = \left(\sum_{\mu=1}^{d} A_\mu(k)_{ab} \, \di k_\mu \right)_{1 \le a,b \le m}, \quad A_\mu(k)_{ab} := -\iu \scal{\psi_a(k)}{\partial_\mu \psi_b(k)}.
\end{equation}
This is indeed the matrix $1$-form of the Grassmann connection on the Bloch bundle $\mathcal{E}$ (\ie the pullback of the standard connection $\di$ via the obvious inclusion $\mathcal{E} \hookrightarrow \T^d \times \Hf$), subordinated to the local trivialization induced by the choice of the Bloch frame. The \emph{abelian} or \emph{$\U(1)$ Berry connection} is then the trace of the connection matrix, namely
\[ \A := \Tr(A) = \sum_{\mu=1}^{d} \A_\mu(k) \, \di k_\mu, \quad \A_\mu(k) := -\iu \sum_{a=1}^{m} \scal{\psi_a(k)}{\partial_\mu \psi_a(k)}. \]

The \emph{Berry curvature} 2-form is the curvature of the Berry connection, namely
\[ F := \di A - \iu \left[A \; \overset{\wedge}{,} \; A\right] \]
which spells out to
\begin{gather*}
F = \sum_{1 \le \mu < \nu \le d} F_{\mu \nu}(k) \, \di k_\mu \wedge \di k_\nu, \\
F_{\mu \nu}(k) := \partial_\mu A_\nu(k) - \partial_\nu A_\mu - \iu \left( A_\nu \wedge A_\mu - A_\mu \wedge A_\nu \right) 
\end{gather*}
(the wedge product between matrix-valued 1-forms entails also the row-by-column product). Similarly, the \emph{abelian} or \emph{$\U(1)$ Berry curvature} is the trace
\begin{equation} \label{eqn:F=dA}
\mathcal{F} := \Tr(F) = \di \A.
\end{equation}
In terms of the Bloch frame $\Psi$, the curvature $\mathcal{F}$ reads
\[ \mathcal{F} = \sum_{1 \le \mu < \nu \le d} \mathcal{F}_{\mu \nu}(k) \, \di k_\mu \wedge \di k_\nu, \quad \mathcal{F}_{\mu \nu}(k) := 2 \mathrm{Im} \left( \sum_{a=1}^{m} \scal{\partial_\mu \psi_a(k)}{\partial_2 \psi_a(k)} \right). \]
However, even if the Bloch frame is just a local object, the Berry curvature is a \emph{global} one, as it can be expressed directly in terms of the family of projectors: a lenghty but straight-forward computation indeed shows that
\begin{equation} \label{BerryProj}
\mathcal{F}_{\mu \nu}(k) = - \iu \Tr_{\Hf} \big( P(k) \left[ \partial_\mu P(k), \partial_\nu P(k) \right] \big).
\end{equation}

When $d=2$, the above identity allows us to rewrite the Chern number as the integral of the (abelian) Berry curvature, namely
\begin{equation} \label{ChernBerry}
c_1(P) = \frac{1}{2\pi} \int_{\T^2} \mathcal{F} \quad \in \Z
\end{equation}
(compare \eqref{eqn:c1}). Moreover, coming back to the Fu--Kane--Mele index of a time-reversal symmetric family of projectors, we can formulate $\delta(P) \in \Z_2$ through the notions we have just introduced as
\begin{equation} \label{eqn:delta}
\delta(P) := \frac{1}{2\pi} \int_{\T^2_+} \mathcal{F} - \frac{1}{2\pi} \int_{\partial \T^2_+} \A \mod 2
\end{equation}
where $\T^2_+$ denotes the set of points in $\T^2$ with non-negative $k_1$ coordinate \cite{FuKane06, CorneanMonacoTeufel16}. Remember that the Berry connection depends on the choice of a Bloch frame: for the above formula to be well-posed one must require that the Bloch frame be \emph{time-reversal symmetric}, in a sense to be specified in the next Subsection. This point will be discussed further in Section \ref{sec:FKM}.

\begin{remark}[Gauge dependence of Berry connection and curvature]
For future refence, let us notice how the Berry connection and curvature matrices, as well as their abelian versions, change under a change of Bloch gauge. If $\Phi$ and $\Psi$ are related by the gauge transformation $U$ as in \eqref{act}, their connection matrices $A^{\Phi}$ and $A^{\Phi}$ are linked by the equation%
\footnote{An easy way to realize this is the following. The connection matrices $A^{\Psi}_\mu(k)$ and $A^{\Phi}_\mu(k)$ satisfy
\[ \Psi(k) \act A^{\Psi}_\mu(k) = -\iu \partial_\mu \Psi(k), \quad \Phi(k) \act A^{\Phi}_\mu(k) = -\iu \partial_\mu \Phi(k). \]
As by definition we have $\Phi(k) = \Psi(k) \act U(k)$, we obtain
\begin{align*}
\Psi(k) \act ( U(k) A^{\Phi}_\mu(k) ) & = \left( \Psi(k) \act U(k) \right) \act A^{\Phi}_\mu(k) = \Phi(k) \act A^{\Phi}_\mu(k) = -\iu \partial_\mu \Phi(k)  \\
& = -\iu \partial_\mu \left( \Psi(k) \act U(k) \right) = \left( -\iu \partial_\mu \Psi(k) \right) \act U(k) + \Psi(k) \act \left( -\iu \partial_\mu U(k) \right) \\
& = \left( \Psi(k) \act A^{\Psi}_\mu(k) \right) \act U(k) + \Psi(k) \act \left( -\iu \partial_\mu U(k) \right) \\
& = \Psi(k) \act \left( A^{\Psi}_\mu(k) U(k) - \iu \partial_\mu U(k) \right)
\end{align*}
by which we deduce that
\[ U(k) A^{\Phi}_\mu(k) = A^{\Psi}_\mu(k) U(k) - \iu \partial_\mu U(k). \]}
\[ A^{\Phi} = U^{-1} \, A^{\Psi} \, U - \iu U^{-1} \, \di U. \]
Taking the trace of both sides of the above equation we obtain the corresponding relation for the abelian Berry connections, namely
\begin{equation} \label{BerryGauge}
\A^{\Phi}(k) = \A^{\Psi}(k) - \iu \Tr \left( U^{-1} \, \di U \right).
\end{equation}

One can similarly compute that the Berry curvature is a gauge-covariant object, namely
\[ F^{\Phi} = U^{-1} \, F^{\Psi} \, U, 	\]
and consequently the abelian Berry curvature $\mathcal{F}$ is gauge-invariant (namely $\mathcal{F}^{\Phi} = \mathcal{F}^{\Psi}$), as could be deduced already from its expression \eqref{BerryProj} given directly in terms of the projectors $P(k)$.
\end{remark}

\subsection{Obstruction theory}
\label{sec:ObsTh}

Even though \eqref{eqn:c1} and \eqref{eqn:delta} express the Chern number and the Fu--Kane--Mele $\Z_2$ index by means of geometric objects related to the family of projectors (its Berry connection and Berry curvature, specifically), the fact that they indeed compute integers or integers $\bmod\: 2$ is a highly non-trivial statement. In the next Sections, we will deduce this fact by means of \emph{obstruction theory}, a framework which allows to identify both indices as topological obstructions. This method has the advantage of manifesting both the quantization and the topological invariance of both indices, and requires only simple tools from linear algebra and basic topology.

Obstruction theory concerns the existence of a Bloch frame for a family of rank-$m$ projectors $\set{P(k)}_{k \in \R^2}$ satisfying (P$_1$), (P$_2$) and, possibly, (P$_3$), which obeys the same symmetries of the projectors themselves. More specifically, we say that a Bloch frame $\Phi$ for $\set{P(k)}_{k \in \R^2}$ is 
\begin{description}
 \item[(F$_1$)] \emph{continuous} if the map $k \mapsto \Phi(k)$ is a continuous map from $\R^2$ to $\Hf^m$;
 \item[(F$_2$)] \emph{$\tau$-equivariant} if%
 \footnote{The action of any (anti)unitary operator on $\Hf$ is lifted to $\Hf^m$ componentwise.}
 \[ \Phi(k+\lambda) = \tau_\lambda \, \Phi(k) \quad \text{for all } k \in \R^2, \: \lambda \in \Lambda; \]
 \item[(F$_3$)] \emph{time-reversal symmetric} if%
 \footnote{The presence of the reshuffling matrix $\eps$ is needed to make the time-reversal symmetry condition self-consistent. This follows essentially from the fact that the antiunitary operator $\Theta$ defines by restriction a symplectic structure on the invariant subspace $\Ran P(k_\sharp) \subset \Hf$ if $k_\sharp \equiv - k_\sharp \bmod \Lambda$. Notice that in particular the rank $m$ of $P(k)$ must be even under (P$_3$).}
 \[ \Phi(-k) = \Theta \Phi(k) \act \eps \]
 for a skew-symmetric unitary matrix $\eps$. Without loss of generality \cite{Hua44}, it can be assumed that
 \begin{equation} \label{eps=J}
 \eps = \begin{pmatrix} 0 & 1 \\ -1 & 0 \end{pmatrix} \oplus \stackrel{m/2 \text{ times}}{\cdots} \oplus \begin{pmatrix} 0 & 1 \\ -1 & 0 \end{pmatrix}.
 \end{equation}
\end{description}

The above properties in general compete agains each other, as was early realized \cite{Kohmoto85, FuKane06} and as becomes apparent upon observing that a continuous, $\tau$-equivariant (and time-reversal symmetric) Bloch frame would provide a global trivialization of the Bloch bundle as a (time-reversal symmetric) vector bundle. 

The general strategy of obstruction theory consists in considering a continuous, globally defined Bloch frame $\Psi$, and trying to modify it in order to obtain a new Bloch frame $\Phi$ which satisfies also the properties of being $\tau$-equivariant and, possibly, time-reversal symmetric. The input frame $\Psi$ can be constructed by covering $\R^d$ with open balls $B_r(k_j)$, $r>0$, $k_j \in \R^d$, in which $\norm{P(k) - P(k_j)} < 1$, $k \in B_r(k_j)$, and using the Kato--Nagy unitary $U(k;k_j)$, which intertwines $P(k)$ and $P(k_j)$, to extend the choice of an orthonormal basis in the vector space $\Ran P(k_j)$ to a continuous choice of an orthonormal basis $\Psi(k)$ in $\Ran P(k)$ (that is, by definition, to a continuous Bloch frame on $B_r(k_j)$) \cite{Kato66}. An alternative construction makes use of the parallel transport associated to the family of projectors $P(k)$, see \eg \cite{CorneanMonacoTeufel16}. The modification of $\Psi$ into $\Phi$ is performed by successive extensions, first at certain high-symmetry points, then along the edges that connect them, and finally on the whole $\R^2$. We will see that this latter step, from 1-dimensional lines to 2-dimensional faces, is in general topologically obstucted, and that this obstruction is encoded in the vanishing of the Chern number if one requires the Bloch frame $\Phi$ to satisfy  (F$_1$) and (F$_2$) (see Section \ref{sec:Chern}), or in the vanishing of the Fu--Kane--Mele index if one also requires (F$_3$) to hold (see Section \ref{sec:FKM}).

\begin{remark}[Analytic Bloch frames]
The obstruction to the existence of symmetric Bloch frames, being topological in nature, fits well inside the continuous category. However, one may wonder wheter an analytic family of projectors as in (P$_1$) admits a Bloch frame depending \emph{analytically} on $k$ as well. This question is crucial in the study of conduction/insulation properties in crystals via \emph{maximally localized Wannier functions} (see \eg \cite{Wannier review, BrouderPanati07}). There are by now several techniques that are able to construct analytic frames out of continuous ones preserving moreover all the symmetries, for example by convolution with suitable kernels \cite{CorneanHerbstNenciu15, CorneanMonacoTeufel16}. These are all incarnations of the more general \emph{Oka's principle}, which states that in fair generality the obstruction to the triviality of a vector bundle in the continuous category can be lifted to the analytic one \cite{Panati07}.
\end{remark}

\section{The Chern number as a topological obstruction}
\label{sec:Chern}

In this Section we illustrate how the Chern number in \eqref{eqn:c1} encodes the topological obstruction to the existence of a continuous and $\tau$-equivariant Bloch frame for a family of projectors $\set{P(k)}_{k \in \R^2}$ satisfying (P$_1$) and (P$_2$).

\subsection{Reduction to the unit cell}
\label{sec:BZ}

The $\tau$-covariance of the family of projectors allows one to focus on points $k$ lying in the \emph{fundamental unit cell} for the lattice $\Lambda = \Span_\Z\set{e_1, e_2}$, namely
\[ \B := \set{k = k_1 e_1 + k_2 e_2 \in \R^2 : |k_j| \le 1/2, \: 1 \le j \le 2}. \]
Indeed, if one can find a continuous Bloch frame $\Phi$ on $\B$ such that $\Phi(k+\lambda) = \tau_\lambda \Phi(k)$ whenever $k \in \B$ and $\lambda \in \Lambda$ are such that $k+\lambda \in \B$ (a condition to be imposed on the boundary of the fundamental unit cell), then one can enforce $\tau$-equivariance to extend the definition of $\Phi$ to the whole $\R^2$ in a continuous way. Conversely, the restriction $\Phi$ to $\B$ of a continuous, $\tau$-equivariant Bloch frame defined on the whole $\R^2$ satisfies exactly the condition stated above.

\begin{figure}[ht]
\sidecaption[t]
%\centering
\begin{tikzpicture}[scale=3]
\coordinate (E1) at (1,0);
\coordinate (E2) at (.5,.866);
\coordinate (mE1) at (-1,0);
\coordinate (mE2) at (-.5,-.866);
\coordinate (V1) at (-.75,-.433);
\coordinate (V2) at (.25,-.433);
\coordinate (V3) at (.75,.433);
\coordinate (V4) at (-.25,.433);
\filldraw [lightgray] (V1) -- (V2) -- (V3) -- (V4) -- cycle;
\draw (V1) node [anchor = north] {$v_1$} -- (V2) node [anchor = north] {$v_2$} -- (V3) node [anchor = south] {$v_3$} -- (V4) node [anchor = south] {$v_4$} -- cycle;
\fill (V1) circle (.8pt)
      (V2) circle (.8pt)
      (V3) circle (.8pt)
      (V4) circle (.8pt);
\draw [-stealth'] (mE1) -- (E1) node [anchor = west] {$e_1$};
\draw [-stealth'] (mE2) -- (E2) node [anchor = south] {$e_2$};
\draw (-.25,-.433) node [anchor = north west] {$E_1$}
      (.5,0) node [anchor = north west] {$E_2$}
      (.25,.433) node [anchor = south east] {$E_3$}
      (-.5,0) node [anchor = south east] {$E_4$};
\end{tikzpicture}
\caption{The fundamental unit cell $\B$, its vertices and its edges.}
\label{fig:BZ}
\end{figure}
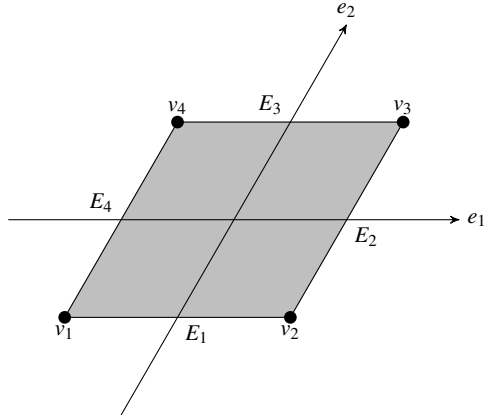

As sketched in Section \ref{sec:ObsTh}, the approach of obstruction theory starts from a Bloch frame $\Psi$ defined on the unit cell. One then modifies its definition on the boundary of $\B$ in order to enforce $\tau$-equivariance there, and then investigates wheter it is possible to extend this modification continuously also on the interior of the unit cell. In particular, this construction on the boundary requires to take care of what happens at the four vertices of $\B$, namely the four points
\[ v_1 = \left(-\frac12, -\frac12\right), \quad v_2 = \left(\frac12, -\frac12\right), \quad v_3 = \left(\frac12, \frac12\right), \quad v_4 = \left(-\frac12, \frac12\right). \]
If this procedure is successful, then the ``output'' frame $\Phi(k)$, $k \in \B$, will satisfy $\tau$-equivariance on the boundary, and it will then be continuously extendable to the whole $\R^2$ by $\tau$-equivariant continuation, as explained above.

Notice that both the input frame $\Psi(k)$ and the output frame $\Phi(k)$ give orthonormal bases for the vector space $\Ran P(k)$, hence they differ by the action of a unitary transformation (a Bloch gauge) $U(k) \in \U(m)$, as in \eqref{BlochGauge}. It is sometimes convenient to consider the continuous map $U \colon \B \to \U(m)$ as the unknown of the problem, rather than the Bloch frame $\Phi$.

We will see that the only step of the construction of $\Phi$ which may be topologically obstructed is the ``face'' extension (from the boundary to the interior of $\B$), and that a quantitative measure of the presence of this topological obstruction is given by the Chern number of the family of projectors.

\subsection{Bloch frame on the boundary}
\label{sec:Boundary}

As a first step, we construct a continuous Bloch frame on the boundary of the fundamental unit cell which satisfies the $\tau$-equivariance condition. The construction can be performed as follows. Given the reference frame $\Psi(v_1)$, one can consider its $\tau$-translates $\tau_{e_1} \Psi(v_1)$ and $\tau_{e_2} \Psi(v_1)$, which constitute orthonormal bases in the subspaces $\Ran P(v_2)$ and $\Ran P(v_4)$, respectively. Let $U\sub{obs}(v_2)$ (respectively $U\sub{obs}(v_4)$) be the unitary matrix which maps the input frame $\Psi(v_2)$ (respectively $\Psi(v_4)$) to $\tau_{e_1} \Psi(v_1)$ (respectively $\tau_{e_2} \Psi(v_1)$):
\[ \tau_{e_1} \Psi(v_1) = \Psi(v_2) \act U\sub{obs}(v_2), \quad \tau_{e_2} \Psi(v_1) = \Psi(v_4) \act U\sub{obs}(v_4). \]
If $\Psi$ were already $\tau$-equivariant then these \emph{obstruction unitaries} would equal the identity matrix. Write $U\sub{obs}(v_\sharp) = \eu^{\iu T(v_\sharp)}$, with $T(v_\sharp) = T(v_\sharp)^*$ self-adjoint, for $v_\sharp \in \set{v_2, v_4}$. Define moreover
\begin{equation} \label{eqn:Phihat}
\widehat{\Phi}(k) := \begin{cases} 
\Psi(k_1,-\frac12) \act \eu^{\iu (2k_1 + 1) T(v_2)/2} & \text{if } k = (k_1, -\frac12) , \: k_1 \in [-\frac12,\frac12], \\
\tau_{e_1} \Psi(-\frac12,k_2) \act \eu^{\iu (2k_2 + 1) T(v_4)/2} & \text{if } k = (\frac12, k_2) , \: k_2 \in [-\frac12,\frac12], \\
\tau_{e_2} \Psi(k_1,-\frac12) \act \eu^{\iu (2k_1 + 1) T(v_2)/2} & \text{if } k = (k_1, \frac12) , \: k_1 \in [-\frac12,\frac12], \\
\Psi(-\frac12,k_2) \act \eu^{\iu (2k_2 + 1) T(v_4)/2} & \text{if } k = (-\frac12, k_2) , \: k_2 \in [-\frac12,\frac12].
\end{cases}
\end{equation}

The frame $\widehat{\Phi}$ is defined on the boundary $\partial \B$ of the fundamental unit cell, where it is also $\tau$-equivariant. Moreover, it is continuous, as on the vertex $v_3$ the definitions coincide. Indeed we have
\[ \tau_{e_1} \Psi(v_4) \act U\sub{obs}(v_4) = \tau_{e_1} \tau_{e_2} \Psi(v_1) = \tau_{e_2} \tau_{e_1} \Psi(v_1) = \tau_{e_2} \Psi(v_2) \act U\sub{obs}(v_2). \]

\subsection{Extension to the face: a topological obstruction}
\label{sec:Extension}

%In Section~\ref{sec:1-skeleton}, we showed that, given a family of projectors satisfying Assumption~\ref{Ass:projectors} and given an associated continuous Bloch frame $\Psi$ on the 2-dimensional unit cell $\B$, we could construct a continuous $\tau$-equivariant Bloch frame $\widehat{\Phi}$ on $\partial \B$. 
In order to see whether it is possible to extend the frame $\widehat{\Phi}$ to a continuous $\tau$-equivariant Bloch frame $\Phi$ defined on the whole unit cell $\B$, we first introduce the unitary map $\widehat{U}(k)$ which maps the input frame $\Psi(k)$ to the frame $\widehat{\Phi}(k)$, \ie such that 
\begin{equation} \label{hatU}
\widehat{\Phi}(k) = \Psi(k) \act \widehat{U}(k), \quad k \in \partial \B
\end{equation}
(compare \eqref{BlochGauge}). This defines a continuous map $\widehat{U} \colon \partial \B \to \U(m)$. If we can find a continuous extension $U \colon \B \to \U(m)$ of $\widehat{U}$ to the unit cell, then $\eqref{BlochGauge}$ can be used to define an extension of the frame $\Phi$ which preserves continuity and $\tau$-equivariance: it turns out that also the converse is true (compare Proposition~\ref{X} below). 

It is a well-known fact in topology \cite[Thm.~17.3.1]{DubrovinNovikovFomenko85} that a continuous map $\widehat{U} \colon \partial \B \to \U(m)$ extends continuously to the inside of the unit cell if and only if the map is homotopically trivial, \ie it can be continuously deformed to a constant map. This condition can be checked by verifying that the integral 
\begin{equation} \label{degdet}
c := \deg([\widehat{U}]) = \frac{\iu}{2 \pi} \oint_{\partial \B} \di k \, \Tr \left( \widehat{U}(k)^{-1} \partial_k \widehat{U}(k) \right)
\end{equation}
vanishes: this is because two maps $\partial \B \to \U(m)$ are homotopic if and only if their \emph{degrees}, defined like in \eqref{degdet}, coincide. Notice that the integral above gives an integer, and provides an isomorphism of the fundamental group $\pi_1(\U(m))$ (whose elements are homotopy classes of maps $\partial \B \to \U(m)$) with the group of integers $\Z$ by assigning $\widehat{U} \mapsto \deg([\widehat{U}])$ \cite[Ch.~8, Sec.~12]{Husemoller94}.

\begin{remark}[Unwinding the determinant is forbidden] \label{rmk:unwinding}
Since we have to extend the \emph{frame} $\widehat{\Phi}$ rather than the \emph{unitary} $\widehat{U}$, one may argue that it may be possible to find another unitary-matrix-valued map that ``unwinds'' the determinant of $\widehat{U}$, while preserving the relevant symmetries of the Bloch frame. This possibility is ruled out by the following result.
\end{remark}

\begin{proposition} \label{X}
Let $\Phi$ be a continuous Bloch frame on $\partial \B$ which is $\tau$-equivariant, and assume that $X \colon \partial \B \to \U(m)$ is a continuous map such that $\Phi \act X$ is also $\tau$-equivariant. Then
\[ \deg([X]) = 0. \]
\end{proposition}
\begin{proof}
We spell out what it means for $\Phi$ and $\Phi \act X$ to be both $\tau$-equivariant:% this yields to
\[ \Phi(k+\lambda) \act X(k+\lambda) = \tau_\lambda \left( \Phi(k) \act X(k) \right) = \tau_\lambda \Phi(k) \act X(k) = \Phi(k+\lambda) \act X(k). \]
This implies that $X(k+\lambda) = X(k)$, whenever $k \in \partial \B$ and $\lambda \in \Lambda$ are such that $k+\lambda \in \partial \B$. As a consequence, the same is true for the expression $x(k) := \Tr \left( X(k)^{-1} \, \partial_k X(k) \right)$ appearing in the integral defining $\deg([X])$ (compare \eqref{degdet}). Denote by $E_i$ the edge of $\partial \B$ connecting $v_i$ with $v_{(i+1)\bmod4}$ (compare Figure~\ref{fig:BZ}). Then the property $x(k+\lambda) = x(k)$ implies that
\[ \int_{E_3} \di k \, x(k) = \int_{-(E_1 + e_1)} \di k \, x(k) = - \int_{E_1} \di k \, x(k), \text{ that is } \int_{E_1+E_3} \di k \, x(k) = 0. \]
Similarly
\[ \int_{E_2+E_4} \di k \, x(k) = 0. \]
We conclude that
\[ \deg([X]) = \frac{\iu}{2 \pi} \, \int_{E_1+E_2+E_3+E_4} \di k \, x(k) = 0 \]
as wanted.
\end{proof}

\subsection{The obstruction is the Chern number}
\label{sec:TopObs}

We now want to rewrite the integer $c$ in \eqref{degdet} and characterize it as a \emph{topological invariant} of the family of projectors $\set{P(k)}_{k \in \R^2}$ (showing in particular that it does not depend on the input Bloch frame $\Psi$ and on the specific interpolation performed on the obstruction matrices in \eqref{eqn:Phihat}). To this end, we will make use of the (abelian) Berry connection and curvature, introduced in Section \ref{sec:Berry}.
If we calculate $\widehat{\A}$ on $\partial \B$ as in \eqref{Berry} using the vectors of the frame $\widehat{\Phi}$ and analogously compute $\A$ using $\Psi$, then
\begin{equation} \label{hatA}
\widehat{\A} = \A - \iu \Tr \left( \widehat{U}^{-1} \di \widehat{U} \right) \quad \text{on } \partial \B,
\end{equation}
in view of \eqref{hatU} and \eqref{BerryGauge}. 
Integrating both sides of Equation \eqref{hatA} on $\partial \B$, we obtain that
\begin{equation} \label{c1-d}
\begin{aligned}
\frac{1}{2\pi} \oint_{\partial \B} \widehat{\A} & = \frac{1}{2\pi} \oint_{\partial \B} \A - \frac{\iu}{2\pi} \oint_{\partial \B} \di k \, \Tr \left( \widehat{U}(k)^{-1} \partial_k \widehat{U}(k) \right) \\
& = \left( \frac{1}{2\pi} \int_{\B} \mathcal{F} \right) - c
\end{aligned}
\end{equation}
by \eqref{eqn:F=dA} and Stokes theorem.

We will now show that the left-hand side of the above equality vanishes. In order to do so, we exploit the $\tau$-equivariance of the Bloch frame $\widehat{\Phi}$, that is, $\widehat{\Phi}(k+\lambda) = \tau_\lambda \widehat{\Phi}(k)$. Indeed, in terms of the Berry connection matrix $A = A(k) \, \di k$ we have that
\begin{equation} \label{eqn:Berrytau}
\begin{aligned}
\widehat{\Phi}(k+\lambda) \act \widehat{A}(k+\lambda) & = - \iu \partial_k \widehat{\Phi}(k+\lambda) = \tau_\lambda \left( - \iu \partial_k \widehat{\Phi}(k) \right) \\
& = \tau_\lambda \left( \widehat{\Phi}(k+\lambda) \act \widehat{A}(k) \right) = \tau_\lambda \widehat{\Phi}(k+\lambda) \act \widehat{A}(k) \\
& =  \widehat{\Phi}(k+\lambda) \act \widehat{A}(k)
\end{aligned}
\end{equation}
%\begin{aligned}
%\widehat{\A}_{\mu}(k+\lambda) & = - \sum_{a=1}^{m} \iu \scal{\widehat{\phi}_a(k+\lambda)}{\partial_{\mu} \widehat{\phi}_a(k+\lambda)} = - \sum_{a=1}^{m} \iu \scal{\tau_\lambda \widehat{\phi}_a(k)}{\tau_\lambda\left(\partial_{\mu} \widehat{\phi}_a(k)\right)} \\
%& = - \sum_{a=1}^{m} \iu \scal{\widehat{\phi}_a(k)}{\partial_{\mu} \widehat{\phi}_a(k)} = \widehat{\A}_{\mu}(k)
%\end{aligned}
%\end{equation}
%since $\tau_\lambda$ is a unitary operator in $\Hi$.
so that $\widehat{A}(k+\lambda) = \widehat{A}(k)$ and, taking the trace, $\widehat{\A}(k+\lambda) = \widehat{\A}(k)$. Arguing similarly to the proof of Proposition~\ref{X}, one can show that the latter relation implies
\[ \int_{E_1+E_3} \widehat{\A} = 0, \quad \int_{E_2+E_4} \widehat{\A} = 0, \]
yielding the vanishing of the left-hand side of \eqref{c1-d}.

Hence we conclude that% \NB{if true, this equality is an interesting byproduct}
\[ \left( \frac{1}{2 \pi} \int_{\B} \mathcal{F} \right) - c = \frac{1}{2\pi} \oint_{\partial \B} \widehat{\A} = 0 \]
which in view of \eqref{ChernBerry} yields
\begin{equation} \label{d=c1}
c  = \frac{1}{2 \pi} \int_{\B} \mathcal{F} = c_1(P)
\end{equation}
as wanted.

\section{The Fu--Kane--Mele invariant as a topological obstruction} \label{sec:FKM}

In this Section, we switch to the time-reversal symmetric setting. As was already mentioned, in this case the presence of a further symmetry kills the topological obstruction given by the Chern number \eqref{eqn:c1} \cite{Panati07, MonacoPanati15}. However, the same symmetry allows to refine the notion of ``symmetric Bloch frame'' by requiring that it be also time-reversal symmetric (compare Section \ref{sec:ObsTh}). This gives rise to a new topological obstruction encoded in the Fu--Kane--Mele $\Z_2$ invariant \cite{FuKane06,FiorenzaMonacoPanati16_F}, as we will now show.

Throughout this Section, $\set{P(k)}_{k \in \R^2}$ denotes a family of orthogonal projectors satisfying (P$_1$), (P$_2$) and (P$_3$).

\subsection{Reduction to the effective unit cell}
\label{sec:EBZ}

In order to investigate the existence of a global Bloch frame for $P(k)$ which is continuous, $\tau$-equivariant, and time-reversal symmetric, it is sufficient to focus one's attention to momenta in the \emph{effective unit cell} for the lattice $\Lambda = \Span_\Z \set{e_1, e_2}$, defined as
\[ \Beff := \set{k = k_1 e_1 + k_2 e_2 \in \R^2 : 0 \le k_1 \le 1/2, \: -1/2 \le k_2 \le 1/2}. \]
Indeed, all points of $\R^2$ can be mapped to $\Beff$ (in an a.e.~unique way) by means of a combination of a translation $k \mapsto k + \lambda$, $\lambda \in \Lambda$, and possibly an inversion $k \mapsto -k$. This means that if a Bloch frame is defined on $\Beff$ and satisfies the relevant symmetries there, then it is possible to extend its definition first to the unit cell $\B$ by enforcing time-reversal symmetry, and secondly to the whole $\R^2$ imposing $\tau$-equivariance. This dictates that the required frame $\Phi$ on $\Beff$ satisfies certain compatibility conditions on the boundary of the effective unit cell, namely that $\Phi(k+\lambda) = \tau_\lambda \Phi(k)$ and $\Phi(-k) = \Theta \Phi(k) \act \eps$, whenever $k \in \partial \Beff$ and $\lambda \in \Lambda$ are such that $\pm k+\lambda \in \partial \Beff$.

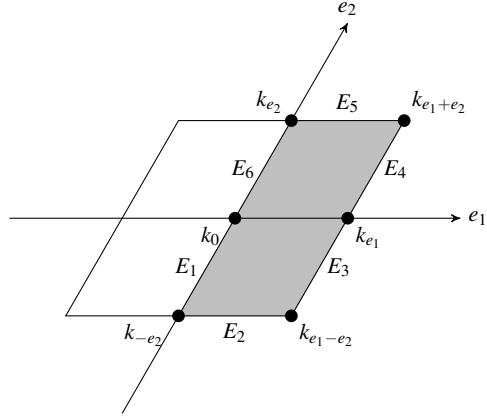
\begin{figure}[ht]
\sidecaption[t]
\centering
\begin{tikzpicture}[scale=3]
\coordinate (E1) at (1,0);
\coordinate (E2) at (.5,.866);
\coordinate (mE1) at (-1,0);
\coordinate (mE2) at (-.5,-.866);
\coordinate (K-E1-E2) at (-.75,-.433);
\coordinate (KE1-E2) at (.25,-.433);
\coordinate (KE1+E2) at (.75,.433);
\coordinate (K-E1+E2) at (-.25,.433);
\coordinate (K0) at (0,0);
\coordinate (KE1) at (.5,0);
\coordinate (KE2) at (.25,.433);
\coordinate (K-E2) at (-.25,-.433);
\filldraw [lightgray] (K-E2) -- (KE1-E2) -- (KE1+E2) -- (KE2) -- cycle;
\fill (K0) node [anchor = north east] {$k_{0}\:\:$} circle (.8pt)
      (KE1) node [anchor = north west] {$k_{e_1}$} circle (.8pt)
      (KE2) node [anchor = south east] {$k_{e_2}$} circle (.8pt)
            (K-E2) node [anchor = north east] {$k_{-e_2}\:\:$} circle (.8pt)
      (KE1-E2) node [anchor = north west] {$k_{e_1-e_2}$} circle (.8pt)
      (KE1+E2) node [anchor = south west] {$k_{e_1+e_2}$} circle (.8pt);
\draw (K-E1-E2) -- (KE1-E2) -- (KE1+E2) -- (K-E1+E2) -- cycle;
\draw [-stealth'] (mE1) -- (E1) node [anchor = west] {$e_1$};
\draw [-stealth'] (mE2) -- (E2) node [anchor = south] {$e_2$};
\draw (-.125,-.217) node [anchor = east] {$E_1$}
      (0,-.433) node [anchor = north] {$E_2$}
      (.375,-.217) node [anchor = west] {$E_3$}
      (.625,.217) node [anchor = west] {$E_4$}
      (.5,.433) node [anchor = south] {$E_5$}
      (.125,.217) node [anchor = east] {$E_6$};
\end{tikzpicture}
\caption{The effective unit cell $\Beff$ and the time-reversal invariant momenta.}
\label{fig:EBZ}
\end{figure}

We will again resort to the technique of obstruction theory. Consequently, we will choose a continuous Bloch frame $\Psi$ on $\Beff$, and try to modify it into a frame $\Phi$ satisfying the symmetries mentioned above. The two frames $\Psi(k)$ and $\Phi(k)$ will be related by a unitary transformation, which we denote by $U(k)$ as in~\eqref{BlochGauge}. As in Section~\ref{sec:BZ}, a special role is played by the high-symmetry points $k_\lambda$, defined by the relation $k_\lambda + \lambda = - k_\lambda$ with $\lambda \in \Lambda$ (that is, $k_\lambda = \lambda/2$). Six such points lie on the boundary of $\Beff$, and are usually referred to as the \emph{time-reversal invariant momenta} (compare Fig.~\ref{fig:EBZ}). 

\subsection{Bloch frame on the boundary}
\label{sec:TRSBoundary}
%Fix a continuous Bloch frame $\Psi$ for $P$, defined on the effective unit cell (Brillouin zone) $\Beff$. In \cite{FiMoPa2}, we showed that it is possible to define a new frame $\widehat{\Phi}$ on $\partial \Beff$ that is continuous, $\tau$-equivariant and time-reversal symmetric.

As a first step, we provide here the construction of a symmetric Bloch frame defined on the boundary of the effective unit cell $\Beff$, following the obstruction-theoretic approach employed in the previous Section for the non-time-reversal-symmetric case.

Let $k_\lambda$ be any of the time-reversal invariant momenta. Given the input frame $\Psi(k_\lambda)$, the transformed frames $\Theta \Psi(k_\lambda) \act \eps$ and $\tau_\lambda \Psi(k_\lambda)$ both give bases of the same vector space $\Ran P(-k_\lambda) = \Ran P(k_\lambda + \lambda)$. As such, they must differ by the action of an \emph{obstruction unitary} matrix:
\begin{equation} \label{eqn:TRSobstruction}
\Theta \Psi(k_\lambda) \act \eps = \tau_\lambda \Psi(k_\lambda) \act U\sub{obs}(k_\lambda).
\end{equation}
These unitary matrices satisfy a further self-compatibility condition, namely
\begin{equation} \label{eqn:Uobseps}
U\sub{obs}(k_\lambda)\tra \, \eps = \eps \, U\sub{obs}(k_\lambda),
\end{equation}
as can be deduced from the following considerations. Applying the operator $\tau_\lambda \Theta = \Theta \tau_{\lambda}^{-1}$ to both sides of the identity \eqref{eqn:TRSobstruction}, and using the defining properties of the time-reversal operator $\Theta$, we obtain
\[ \tau_{\lambda} \Psi(k_\lambda) \act (- \overline{\eps}) = \Theta \Psi(k_\lambda) \act \overline{U\sub{obs}(k_\lambda)}. \]
Using the relation \eqref{eqn:TRSobstruction} again we can rewrite the above equality as
\[ \Theta \Psi(k_\lambda) \act \left( -\eps \, U\sub{obs}(k_{\lambda})^{-1} \, \overline{\eps}\right) = \Theta \Psi(k_\lambda) \act \overline{U\sub{obs}(k_\lambda)} \]
from which we deduce that $-\eps \, U\sub{obs}(k_{\lambda})^{-1} \, \overline{\eps} = \overline{U\sub{obs}(k_\lambda)}$. Taking complex conjugates and using the fact that $-\overline{\eps}=\eps^{-1}$ (by unitarity and skew-symmetry) yields exactly \eqref{eqn:Uobseps}.

Write now $U\sub{obs}(v_\sharp) = \eu^{\iu T(v_\sharp)}$ for $v_\sharp \in \set{v_1, \ldots, v_4}$, with $T(v_\sharp) = T(v_\sharp)^*$ self-adjoint and satisfying $\sigma(T(v_\sharp)) \subset (- \pi, \pi]$. This normalization on the arguments of the eigenvalues of $U\sub{obs}(v_\sharp)$ gives that $T(v_\sharp)$ inherits the property \eqref{eqn:Uobseps} in the form
\begin{equation} \label{eqn:Teps}
T(v_\sharp)\tra \, \eps = \eps \, T(v_\sharp)
\end{equation}
(see \cite[Sec.~6, Lemma]{Hua44}).

Set now
\begin{equation} \label{eqn:TRSPhihat}
\widehat{\Phi}(k) := \begin{cases} 
\Psi(k) \act V(k) & \text{if } k \in S, \\
\tau_{e_1}^{-1} \Theta \Psi(\frac12,-k_2) \act \left( \overline{V(\frac12,-k_2)} \, \eps \right) & \text{if } k =(\frac12, k_2), \: k_2 \in [0,\frac12], \\
\tau_{e_2} \Psi(k_1, -\frac12) \act V(k_1,-\frac12) & \text{if } k = (k_1, \frac12), \: k_1 \in [0,\frac12], \\
\Theta \Psi(0,-k_2) \act \left( \overline{V(0,-k_2)} \, \eps \right) & \text{if } k =(0, k_2), \: k_2 \in [0,\frac12],
\end{cases}
\end{equation}
where 
\begin{multline*}
S := \set{k = \left(0,k_2\right): k_2 \in \left[-\tfrac12,0\right]} \cup \set{k = \left(k_1,-\tfrac12\right): k_1 \in \left[0,\tfrac12\right]} \\
\cup \set{k = \left(\tfrac12,k_2\right): k_2 \in \left[-\tfrac12,0\right]}
\end{multline*}
and for $k \in S$
\begin{equation} \label{eqn:V(k)}
V(k) := \begin{cases}
\eu^{\iu [(1+2k_2) T(v_1) - 2 k_2 T(v_2)]/2} & \text{if } k=(0,k_2), \: k_2 \in [-\frac12,0], \\
\eu^{\iu [(1-2k_1) T(v_2) + 2 k_1 T(v_3)]/2} & \text{if } k=(k_1,-\frac12), \: k_1 \in [-\frac12,0], \\
\eu^{\iu [(1+2k_2) T(v_3) - 2 k_2 T(v_4)]/2} & \text{if } k=(\frac12,k_2), \: k_2 \in [-\frac12,0].
\end{cases}
\end{equation}
Equation~\eqref{eqn:TRSPhihat} above defines a Bloch frame $\widehat{\Phi}$ on $\partial \Beff$ which is by construction $\tau$-equivariant and time-reversal symmetric. Notice also that \eqref{eqn:V(k)} yields
\[ U\sub{obs}(k_\lambda) = V(k_\lambda)^2 = V(k_\lambda) \eps^{-1} V(k_\lambda)\tra \eps\]
at the time-reversal invariant momenta. Repeated use of the defining property \eqref{eqn:TRSobstruction} for $U\sub{obs}(k_\lambda)$ and of its generator $T(k_\lambda)$, together with \eqref{eqn:Uobseps} and \eqref{eqn:Teps}, shows that $\widehat{\Phi}$ also joins continuously at the time-reversal invariant momenta. For example, at $k_\lambda = k_{e_1} = (1/2,0)$ we have
\begin{align*}
\tau_{e_1}^{-1} \Theta \Psi(k_{e_1}) \act \left( \overline{V(k_{e_1})} \, \eps \right) & = \tau_{e_1}^{-1} \Theta \Psi(k_{e_1}) \act \left( \eps \, V(k_{e_1})^* \right) \\
& = \tau_{e_1}^{-1} \left( \Theta \Psi(k_{e_1}) \act \eps \right) \act \, V(k_{e_1})^* \\
& = \tau_{e_1}^{-1} \left( \tau_{e_1} \Psi(k_{e_1}) \act U\sub{obs}(k_{e_1}) \right) \act V(k_{e_1})^{-1} \\
& = \Psi(k_{e_1}) \act\left( V(k_{e_1})^2 V(k_{e_1})^{-1} \right) = \Psi(k_{e_1}) \act V(k_{e_1}).
\end{align*}

\subsection{Extension to the face: a topological obstruction}
\label{sec:TRSExtension}

Let $\widehat{U}$ denote the unitary transformation mapping the input frame $\Psi$ to the Bloch frame $\widehat{\Phi}$ we just constructed, as in \eqref{hatU}. We have already argued in the previous Section that the obstruction to the continuous extension of the map $\widehat{U} \colon \partial \Beff \to \U(m)$ to the interior of the effective unit cell is measured precisely by the vanishing of the integer $\deg([\widehat{U}]) \in \Z$ given by
\begin{equation} \label{eqn:TRSdefdeg}
\deg([\widehat{U}]) = \frac{\iu}{2 \pi} \oint_{\partial \Beff} \di k \, \Tr \left( \widehat{U}(k)^{-1} \partial_k \widehat{U}(k) \right) 
\end{equation}
(compare \eqref{degdet}). However, in this new setting it is no longer the case that the extension problem for the unitary $\widehat{U}$ is equivalent to the one for the frame $\widehat{\Phi}$, as opposed to the situation in Remark~\ref{rmk:unwinding}. Indeed, we have the following result.

\begin{proposition} \label{TRSX}
Let $\Phi$ be a continuous Bloch frame on $\partial \Beff$ which is symmetric, and assume that $X \colon \partial \Beff \to \U(m)$ is a continuous map such that $\Phi \act X$ is also symmetric. Then
\[ \deg([X]) \in 2 \Z. \]
\end{proposition}
\begin{proof}
One easily computes that asking that $\Phi \act X$ be again $\tau$-equivariant and time-reversal symmetric is equivalent to the conditions
\begin{equation} \label{eqn:propX}
X(k+\lambda) = X(k) \quad \text{and} \quad X(-k+\lambda)\tra \, \eps \, X(k) = \eps
\end{equation}
whenever $k \in \partial \Beff$ and $\lambda \in \Lambda$ are such that $\pm k+\lambda \in \partial \Beff$. In view of the above conditions, the integral computing the degree of $X$, as in \eqref{eqn:TRSdefdeg}, simplifies to
\begin{equation} \label{TRSdegX}
\deg([X]) = 2 \left\{ \frac{\iu}{2 \pi} \int_{E_1} \di k \, \Tr \left( X(k)^{-1} \partial_k X(k) \right) + \frac{\iu}{2 \pi} \int_{E_3} \di k \, \Tr \left( X(k)^{-1} \partial_k X(k) \right) \right\}
\end{equation}
where the $E_i$'s are the portions of $\partial \Beff$ connecting two consecutive time-reversal invariant momenta (compare Fig.~\ref{fig:EBZ}).

Notice now that for a unitary-matrix-valued map
\[ \Tr \left( X(k)^{-1} \partial_k X(k) \right) = \xi(k)^{-1} \partial_k \xi(k), \quad \text{with } \xi(k) = \det X(k) \in U(1) \]
(see \eg \cite[Lemma~2.12]{CorneanMonacoTeufel16}). On $E_1$ and $E_3$, the maps $k \mapsto \xi(k)$ are actually periodic, since the second condition in \eqref{eqn:propX} implies that at the time-reversal invariant momenta $k_\lambda$ the matrix $X(k_\lambda)$ must be symplectic and thus of unit determinant. The term in curly brackets on the right-hand side of \eqref{TRSdegX} then computes the sum of the winding numbers of the maps $\xi\big|_{E_1}$ and $\xi\big|_{E_3}$, and is thus an integer. This concludes the proof of the Proposition.
\end{proof}

The above result shows that if $\deg([\widehat{U}]) = 2 r \in 2 \Z$ is even, it is still possible to ``unwind'' the map $\widehat{U}$ with the help of an auxiliary map $X$, without breaking the symmetries ($\tau$-equivariance, time-reversal) enjoyed by the frame $\widehat{\Phi}$ as in \eqref{eqn:TRSPhihat}. Indeed, it is easily verified that the map $X \colon \partial \Beff \to \U(m)$ defined (in the basis where $\eps$ is of the form \eqref{eps=J}) by
\[ X(k) = \begin{cases}
\eu^{-2 \pi \iu r (k_2 + 1/2)} \Id_2 \oplus \Id_{m-2} & \text{if } k = (\frac12, k_2) \in E_3 \cup E_4, \: k_2 \in [-\frac12,\frac12], \\
\Id_m & \text{elsewhere in } \partial \Beff,
\end{cases} \]
satisfies \eqref{eqn:propX} and $\deg([X]) = - 2r$. It follows that the frame $\Psi \act (\widehat{U} X)$ is still continuous, $\tau$-equivariant and time-reversal symmetric, and extends to a continuous Bloch frame $\Phi$ in the interior of $\Beff$ since $\deg([\widehat{U} X])=0$.

We conclude that the topological obstruction to the existence of a continuous and symmetric Bloch frame is measured by the quantity 
\begin{equation} \label{delta}
d := \deg([\widehat{U}]) \mod 2.
\end{equation}
It can be shown \cite{FiorenzaMonacoPanati16_F} that $d \in \Z_2$ defines a true \emph{topological invariant} for the family of projectors $P(k)$ enjoying (P$_1$), (P$_2$) and (P$_3$), that is, it does not depend on the choice of the input frame $\Psi$ and on the explicit form of the interpolation $V$ as in \eqref{eqn:V(k)}, and moreover it stays constant under continuous deformations (homotopies) of the family of projectors which preserve its symmetry properties.

\subsection{The obstruction is the Fu--Kane--Mele index}
\label{sec:TRSTopObs}

Arguing along the same lines of Section~\ref{sec:TopObs}, it is possible to write the topological invariant $d$ in terms of the Berry connection and the Berry curvature associated to the family of projectors $P(k)$, which in turn connects $d$ with the Fu--Kane--Mele $\Z_2$ index for time-reversal symmetric topological insulators \cite{FuKane06}. Indeed, the analogue of \eqref{c1-d} reads now
\begin{equation} \label{eqn:delta=FKM}
\delta = \frac{1}{2 \pi} \int_{\Beff} \mathcal{F} - \frac{1}{2\pi} \oint_{\partial \Beff} \widehat{\A} \mod 2 = \delta(P),
\end{equation}
compare \eqref{eqn:delta}. Let us stress that the Berry connection $\widehat{\A}$ appearing in the above formula must be computed with respect to a frame $\widehat{\Phi}$ which is $\tau$-equivariant \emph{and time-reversal symmetric} on the boundary of the effective unit cell $\Beff$. This guarantees, for example, that the expression on the right-hand side is gauge-independent, as we have seen how a change of unitary gauge which preserves the symmetries must have even degree (Proposition~\ref{TRSX}). 

Notice that, contrary to the case of the Chern number treated in Section~\ref{sec:TopObs}, the ``boundary term'' in \eqref{eqn:delta=FKM} need not vanish. Indeed, the symmetries
\begin{equation} \label{Berrysymm} 
\widehat{\A}(k+\lambda) = \widehat{\A}(k) = \widehat{\A}(-k)
\end{equation}
of the coefficient of the Berry connection $1$-form, which are inherited from the $\tau$-equivariance and the time-reversal symmetry of the underlying frame $\widehat{\Phi}$, only imply that its integral on $\partial \Beff$ can be simplified to
\[ \oint_{\partial \Beff} \widehat{\A} = 2 \int_{E_1+E_3} \widehat{\A} = 2  \left( \int_{0}^{1/2} \di k_2 \, \left[ \widehat{\A}(1/2,k_2) -  \widehat{\A}(0,k_2)\right] \right). \] 
The first equality in \eqref{Berrysymm} can be argued exactly as in \eqref{eqn:Berrytau}, while for the second one we proceed as follows. From the time-reversal symmetry property $\widehat{\Phi}(-k) = \Theta \widehat{\Phi}(k) \act \eps$ it follows that
\[ \Theta \left( \partial_k \widehat{\Phi}(k) \right) \act \eps = - \left( \partial_k \widehat{\Phi} \right)(-k). \]
Using the relation above together with the defining property $\widehat{\Phi}(k) \act \widehat{A}(k) = -\iu \partial_k \widehat{\Phi}(k)$ for the connection matrix $\widehat{A}(k)$ we then obtain
\begin{align*}
\Theta \widehat{\Phi}(k) \act \left( \eps \widehat{A}(-k) \right) & = \widehat{\Phi}(-k) \act \widehat{A}(-k) = - \iu \left( \partial_k \widehat{\Phi} \right)(-k) \\
& = \iu \Theta \left( \partial_k \widehat{\Phi}(k) \right) \act \eps = \Theta \left(- \iu \partial_k \widehat{\Phi}(k) \right) \act \eps \\
& = \Theta \left( \widehat{\Phi}(k) \act \widehat{A}(k) \right) \act \eps = \Theta \widehat{\Phi}(k) \act \left( \overline{\widehat{A}(k)} \eps \right).
\end{align*}
We conclude that $\widehat{A}(-k) = \eps^{-1} \overline{\widehat{A}(k)} \eps$, and by taking the trace that $\widehat{\A}(-k) = \overline{\widehat{\A}(k)} =\widehat{\A}(k)$, because $\widehat{\A}(k)$ takes values in the Lie algebra $\mathfrak{u}(1) = \R$.

\begin{acknowledgement}
The author is indebted with D.~Fiorenza and G.~Panati for valuable discussions, as well as with G.~dell'Antonio and A.~Michelangeli for the invitation to the INdAM meeting ``Contemporary Trends in the Mathematics of Quantum Mechanics''. Financial support from INdAM and from the German Science Foundation (DFG) within the GRK 1838 ``Spectral theory and dynamics of quantum systems'' is gratefully acknowledged.
\end{acknowledgement}
%
%%\section*{Appendix}
%%\addcontentsline{toc}{section}{Appendix}
%
%
%When placed at the end of a chapter or contribution (as opposed to at the end of the book), the numbering of tables, figures, and equations in the appendix section continues on from that in the main text. Hence please \textit{do not} use the \verb|appendix| command when writing an appendix at the end of your chapter or contribution. If there is only one the appendix is designated ``Appendix'', or ``Appendix 1'', or ``Appendix 2'', etc. if there is more than one.

%\input{referenc}

\end{document}